\documentclass[10pt]{article}

\usepackage[applemac]{inputenc}

\usepackage{microtype}	
\usepackage[siunitx]{circuitikz}
\usepackage{showexpl}
\usepackage{framed}
\usepackage{hyperref}
\usepackage[fleqn]{amsmath}
\usepackage{amsmath,amsthm,amsfonts,amscd,bezier,caption,braket,MnSymbol,mathtools,comment,graphicx}
\DeclareMathOperator{\tr}{Tr}
\DeclareMathOperator{\Ker}{Ker}
\DeclareMathOperator{\IM}{Im}
\DeclareMathOperator{\rank}{rank}
\renewcommand{\=}{:=}

\newcommand{\RR}{\mathbb{R}}
\newcommand{\pT}{\partial^{T}}

\newcommand{\de}{\delta}
\newcommand{\p}{\partial}
\newtheorem{thm}{Theorem}[section]

\newtheorem{cor}[thm]{Corollary}

\theoremstyle{definition}

\theoremstyle{definition}

\theoremstyle{definition}
\newtheorem{rem}[thm]{Remark}

\theoremstyle{definition}

\numberwithin{equation}{section}
\numberwithin{table}{section}
\begin{document}
\title{\LARGE\bf Note on homological modeling\\ of the electric circuits}
\date{}
\author{Eugen Paal and M\"art Umbleja}

\maketitle

\thispagestyle{empty}

\begin{abstract}
Based on a simple example, it is explained how the homological analysis may be applied for modeling of the electric circuits. The homological \emph{branch}, \emph{mesh} and \emph{nodal} analyses are presented. Geometrical interpretations are given.
\end{abstract}

\section{Introduction and outline of the paper}

The classical electric circuit analysis is based on the 2 Kirchhoff Laws \cite{O'Malley}:

\begin{enumerate}
\itemsep-2pt
\item{[KCL]}
\textit{Kirchhoff's current law} says that: At any instant in a circuit the algebraic sum of the currents entering a node equals the algebraic sum of those leaving.
\item{[KVL]}
\textit{Kirchhoff's voltage law} says that: At any instant around a loop, in either a clockwise or counterclockwise direction, the algebraic sum of the voltage drops equals the algebraic sum of the voltage rises.
\end{enumerate}

The homological analysis of the electric circuits is based on its geometric elements - nodes, contours (edges, branches), meshes (simple closed loops), also called the chains, and using the geometric boundary operator of the circuit. The latter depends only on the geometry (topology) of the circuit. Then, both of the Kirchhoff laws can be presented in a compact algebraic form called the homological Kirchhoff Laws (HKL).

In the present note, based on a simple example, it is explained how the homological analysis may be applied for modeling of the electric circuits. The homological \emph{branch}, \emph{mesh} and \emph{nodal} analyses are presented. With slight modifications, we follow \cite{BS90,Frankel,Roth55} where the reader can find more involved theoretical details and related References as well. Geometrical interpretations are given. For simplicity, the cohomological aspects are not exposed.

\section{Notations}

Consider a simple DC electric circuit $C$ on Fig \ref{circuit}. It has the following basic geometric spanning spaces:
\begin{figure}[h]
\caption{DC circuit}
\label{circuit}
\begin{center}
\resizebox{!}{5.3cm}{\includegraphics{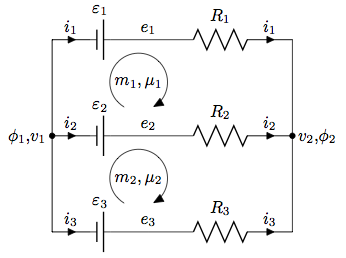}} 
\end{center}
\end{figure}
%
\begin{itemize}
\itemsep-2pt
\item
\emph{Node space} $C_0\=\braket{v_1 v_2}_\RR$,
\item
\emph{Contour space} $C_1\=\braket{e_1 e_2 e_3}_\RR$,
\item
\emph{Mesh space} $C_2\=\braket{m_1 m_2}_\RR$. 
\end{itemize}
Elements of $C_n$ are called $n$-chains and we denote $C\=(C_n)_{n=0,1,2}$. Denote the algebraic electrical parameters (with values from the coefficient field $\RR$) as follows:

\begin{itemize}
\itemsep-2pt
\item
$\phi_1,\phi_2$ - the node potentials,  
\item
$i_1,i_2,i_3$ - the contour currents,
\item 
$\varepsilon_1,\varepsilon_2,\varepsilon_3$ - the voltages,
\item
$\mu_1,\mu_2$ - the mesh currents.
\item
$\RR_+ \ni R_1,R_2,R_3$ - the resistors,
\item
$\RR_+ \ni G_1,G_2,G_3$ - the conductances, defined by $G_nR_n\=1$ ($n=1,2,3$).
\end{itemize}
To denote the physical variables, it is convenient to use the Dirac \textit{bra-ket} notations. Thus, denote the rows by \emph{bra}-vectors, e.g,
\begin{align}
\label{notations}
\bra{\phi}
&\=\bra{\phi_1\phi_2}\=[\phi_1\phi_1]\\
\bra{i}
&\=\bra{i_1 i_2 i_3}\=[i_1 i_2 i_3]\\
\bra{\mu}
&\=\bra{\mu_1 \mu_2}\=[\mu_1 \mu_2]\\
\bra{\varepsilon}
&\=\bra{\varepsilon_1 \varepsilon_2 \varepsilon_3}\=[\varepsilon_1 \varepsilon_2 \varepsilon_3]\\
\bra{R}
&\=\bra{R_1 R_2 R_3}\=[R_1 R_2 R_3]\\
\bra{G}
&\=\bra{G_1 G_2 G_3}\=[G_1 G_2 G_3]
\end{align}
and their (here real) transposes are denoted by ket-vectors $\ket{\cdots}\=\bra{\cdots}^{T}$, the latter are thus \emph{columns}. In such a notation, the bra-ket vectors may be considered as coordinate vectors of the chains.  One must be careful about context, i.e the physical meaning of the bra-kets, the chain spaces must be distinguished according to the physical units. Also, not all chains represent the \emph{physical} states.

\section{Circuit metrics \& scalar product}

By definition, the circuit metrical matrix is symmetric, \emph{positively} defined and reads
\begin{align}
R 
=
\begin{bmatrix*}[r]
R_{1} & 0 & 0\\
0 & R_{2} & 0\\
0 & 0 & R_{3}
\end{bmatrix*},
\quad \det R>0,\quad R^{T}=R\quad \text{(symmetry)}
\end{align}
and its inverse, defined by  $GR=1_{3\times3}=RG$ is
\begin{align}
G
\=
G
=
\dfrac{1}{R_{1} R_{2} R_{3}}
\begin{bmatrix*}[c]
R_{2}R_{3} & 0 & 0\\
0 & R_{1}R_{3} & 0\\
0 & 0 & R_{1}R_{2}
\end{bmatrix*}
=
\begin{bmatrix*}[c]
G_{1} & 0 & 0\\
0 & G_{2} & 0\\
0 & 0 & G_{3}
\end{bmatrix*}
=G^{T}
\end{align}
Let $\rho$ denote either $R$ or $G$. Then the \emph{non-euclidean} elliptic (iso)scalar product $\braket{\cdot|\cdot}_\rho$ is defined by $\braket{\cdot|\cdot}_\rho \= \braket{\cdot|\rho|\cdot}$. One must be careful with limits (contractions) $\det \rho\to0$ and remember that every \emph{physical} wire and voltage source has at least its (nontrivial) \emph{positive self-resistance}, that may be included in $\rho$, so that $\det\rho>0$. 

With respect to the circuit metrics we may define the (iso)\emph{norm} function $|\cdot|_\rho$ by $|x|_\rho\=\sqrt{\braket{x|x}_\rho}$. Then, the Cauchy-Schwartz (CS)  inequality $|\braket{x|y}_\rho |\leq |x|_\rho |y|_\rho$ is evident whenever the scalar product exists for given vectors, as CS inequality holds for every scalar (inner) product.

\section{Boundary operator}

Now construct the boundary operator $\p\=(\p_n)_{n=0,1,2,3}$ of the electric circuit  presented on Fig. \ref{circuit} and its matrix representation. In what follows, we identify the chains with their coordinate ket-vectors. 

First construct $\p_0:C_0\to C_{-1}\=\braket{0}_\RR$. By definition, the nodes (vertices) are elementary elements of circuits with trivial boundaries, thus 
\begin{align}
\p_0 v_1\=0\=\p_0 v_2
\quad\Longrightarrow\quad
\p_0
=[0\quad 0]\=0_{1\times2}
\end{align}
Next define $\p_1:C_1\to C_0$, which acts on the directed contours (edges, branches) by
\begin{align}
\p_1 e_1 \= v_2-v_1 \= \ket{-1;1}\\
\p_1 e_2 \= v_2-v_1 \= \ket{-1;1}\\
\p_1 e_3 \= v_2-v_1 \= \ket{-1;1}
\end{align}
In coordinate (matrix) representation one has
\begin{align}
\p_1 =
\begin{bmatrix*}[r]
-1 & -1 & -1\\
1  & 1 & 1
\end{bmatrix*}
\quad\Longrightarrow\quad
\pT_{1} =
\begin{bmatrix*}[r]
-1 & 1\\
-1 & 1\\
-1 & 1
\end{bmatrix*}
\=
\delta^{0}
\end{align}
Evidently, $\p_0\p_1=0$. Note that $\rank \p_1=1$. 

Now define $\p_2:C_2\to C_1$, which acts on the closed clockwise directed contours by
\begin{align}
\p_2 m_1 \= e_1-e_2 \= \ket{1;-1;0}\\
\p_2 m_2 \= e_2-e_3 \= \ket{0;1;-1}
\end{align}
from which it follows that
\begin{align}
\p_2
\=
\begin{bmatrix*}[r]
1 & 0\\
-1 & 1\\
0 & -1
\end{bmatrix*}
\quad\Longrightarrow\quad
\pT_2
=
\begin{bmatrix*}[r]
1 & -1 & 0\\
0 & 1 & -1
\end{bmatrix*}
\=
\delta^{1}
\end{align}
One again can easily check that $\p_1\p_2=0_{2\times2}$ as well as $\delta^{1}\delta^{0}=0$. Note that $\rank\p_2=2$.  

We finalize the construction by defining $\p_3\=0_{2\times1}$, which means that $C_3\=\braket{0}_\RR$.
\begin{rem}
One must be careful when comparing our representation with \cite{Roth55,BS90,Frankel}, where the mesh space $C_2$ is identified with $\IM\p_2\subset C_1$.
\end{rem}

\section{Homology}

The boundary operator is defined by its action on the geometrical elements of the circuit, thus not depending on the particular electrical parameters - potentials, voltages, circuit and mesh currents - but only on the topology of the circuit under consideration. One can visualize the boundary operator and its (mathematical) domains and codomains by the following complex: 
\begin{align}
\label{e-sequence}
\begin{CD}
(C_3\=)\quad 0 @> (\p_3=0)  >> C_2 @> \p_2 >> C_1 @> \p_1 >> C_0 @> (\p_0=0)  >> 0 \quad (\=C_{-1})
\end{CD}
\end{align}
As we have seen, the boundary operator is \emph{nilpotent},
\begin{align}
\label{nilpotency}
\IM\p_{n+1}\subseteq\Ker\p_n
\quad\Longleftrightarrow\quad
\p_{n}\p_{n+1}=0,\quad n=0,1,2
\end{align}
which is concisely denoted as $\p^{2}=0$.  A complex  $(C,\p)\=(C_n,\p_n)_{n=0,1,2,3}$ is said to be \emph{exact} at $C_n$ if $\IM\p_{n+1}=\Ker\p_n$. To study the  \emph{exactness} of this complex, the following three conditions must be inquired:
\begin{enumerate}
\itemsep-2pt
\item
$0\overset{?}{=}\Ker\p_2$
\item
$\IM\p_2\overset{?}{=}\Ker\p_1$
\item
$\IM\p_1\overset{?}{=}C_0$
\end{enumerate}
For a short exact sequence one can write:
\begin{align}
C_0\cong\dfrac{C_1}{\Ker\p_1}\cong\dfrac{C_1}{\IM\p_2},
\quad
\dim C_1=\dim C_0+\dim\Ker\p_1=\dim C_0+\dim\IM\p_2
\end{align}
The deviation of a complex from exactness can be described by the \emph{homology} concept. The \emph{homology} of the complex $C\=(C_n,\p_n)_{n=0,1,2,3}$ is the sequence $H(C)\=(H_n(C))_{n=0,1,2}$ with homogeneous components $H_n(C)$ called the \emph{homology spaces} that are defined as quotient spaces
\begin{align}
\label{hom-spaces}
H_{n}(C)
\=
\dfrac{\!\!Z_n (C)\=\Ker\p_n}{B_n(C)\=\IM\p_{n+1}},\quad \dim Z_n=\dim H_n+\dim B_n, \quad n=0, 1, 2
\end{align}
Chains from $Z(C)\=\Ker \p$ are called \emph{cycles} and from $B(C)\=\IM\p$ \emph{boundaries}. 

Note that correctness of the homology construction is based on the inclusion \eqref{nilpotency}. One can easily see from \eqref{hom-spaces} that in homological terms the exactness conditions may be presented as follows:
\begin{enumerate}
\itemsep-2pt
\item
$0=\Ker\p_2\hspace{10mm}\Longleftrightarrow\quad H_2=0\quad\Longleftrightarrow\quad \dim H_2=0$
\item
$\IM\p_2=\Ker\p_1\quad\Longleftrightarrow\quad H_1=0\quad\Longleftrightarrow\quad \dim H_1=0$
\item
$\IM\p_1=C_0\hspace{9mm}\Longleftrightarrow\quad H_0=0\quad\Longleftrightarrow\quad \dim H_0=0$
\end{enumerate}

\section{Homological Kirchhoff Laws}

As we can see, not all chains represent the \emph{physical} states. The real electrical configurations are prescribed by the Kirchhoff Laws. The homological (form of the) Kirchhoff Laws (HKL) read (see e.g \cite{Roth55,BS90,Frankel})
\begin{enumerate}
\itemsep-2pt
\item{[HKCL]}
$\p_1\ket{i}=0$
\hspace{13mm}\quad $\Longleftrightarrow$ \quad 
$\ket{i}\in\Ker\p_1$
\item{[HKVL]}
$R\ket{i}=\ket{\varepsilon}-\delta^{0}\ket{\phi}$
\quad $\Longleftrightarrow$ \quad 
$R\ket{i}-\ket{\varepsilon}\in\IM\de^{0}$
\end{enumerate}
Thus, in homological terms, the Kirchhoff Laws can compactly be presented  by using the boundary (and coboundary) operators of a particular circuit. 

It must be noted that the Kirchhoff Laws are the \emph{physical} laws and as other physical laws these can not be fully proved mathematically or other theoretical discussions, but tested only via the physical measurements and observations. This concerns the HKL as well. Here one can observe certain analogy with variational principles of physics and the differential equations of the dynamical systems.

Below we present the homological \emph{branch}, \emph{mesh} and \emph{nodal} analyses, as well as explain how the KCL and KVL follow from the HKL. 

The HKCL tells us that the \emph{physical} branch currents are realized only in $Z_1(C)\=\Ker\p_1$ (cycles). To describe the latter, recall the circuit notations \eqref{notations}. We have
\begin{align}
\ket{\varepsilon}-\de^{0}\ket{\phi}
=
\begin{bmatrix*}[r]
\varepsilon_{1} \\
\varepsilon_{2} \\
\varepsilon_{3}
\end{bmatrix*}
-
\begin{bmatrix*}[r]
-1 & 1\\
-1 & 1\\
-1 & 1
\end{bmatrix*}
\begin{bmatrix*}[r]
\phi_{1} \\
\phi_{2}
\end{bmatrix*}
= 
\begin{bmatrix*}[r]
\varepsilon_{1}+\phi_{1}-\phi_{2}\\
\varepsilon_{2}+\phi_{1}-\phi_{2}\\
 \varepsilon_{3}+\phi_{1}-\phi_{2}
\end{bmatrix*}
\end{align}
and one can see that
\begin{align}
\ket{i}
\=&
\ket{i_1 i_2 i_3}\\
=&
G \left(\ket{\varepsilon}-\de^{0}\ket{\phi} \right)
\\
=&
\begin{bmatrix*}[r]
G_{1}&0 & 0\\
0 & G_{2} & 0\\
0 & 0 & G_{3}
\end{bmatrix*}
\begin{bmatrix*}[r]
\varepsilon_{1}+\phi_{1}-\phi_{2}\\
\varepsilon_{2}+\phi_{1}-\phi_{2}\\
\varepsilon_{3}+\phi_{1}-\phi_{2}
\end{bmatrix*}\\
=&
\begin{bmatrix*}[r]
G_{1}(\varepsilon_{1}+\phi_{1}-\phi_{2})\\
G_{2}(\varepsilon_{2}+\phi_{1}-\phi_{2})\\
G_{3}(\varepsilon_{3}+\phi_{1}-\phi_{2})
\end{bmatrix*} \\
=&
\ket{\underbrace{G_{1}(\varepsilon_{1}+\phi_{1}-\phi_{2}}_{i_1});
\underbrace{G_{2}(\varepsilon_{2}+\phi_{1}-\phi_{2}}_{i_2});
\underbrace{G_{3}(\varepsilon_{3}+\phi_{1}-\phi_{2}}_{i_3})}
\end{align}
which results
\begin{align}
i_{1}
=
\dfrac{\varepsilon_{1}+\phi_{1}-\phi_{2}}{R_{1}},
\quad
i_{2}
=
\dfrac{\varepsilon_{2}+\phi_{1}-\phi_{2}}{R_{2}},
\quad
i_{3}
=
\dfrac{\varepsilon_{3}+\phi_{1}-\phi_{2}}{R_{3}}
\end{align}
Now, the conventional KVL around the closed loops easily follow:
\begin{align}
\bar{\de}\phi
\=
\phi_1-\phi_2
&= -\varepsilon_1+R_1 i_1\\
&= -\varepsilon_2+R_2 i_2\\
&= -\varepsilon_3+R_3 i_3
\end{align}
while consistency is evident. Hence, the HKCL reads
\begin{align}
\label{i-plane}
\p_{1}\ket{i_1i_2i_3}=0
\quad\Longleftrightarrow\quad i_{1}+i_{2}+i_{3}=0
\quad \text{(branch currents plane)}, \quad \dim\Ker\p_{1}=2
\end{align}
that we can rewrite as
\begin{align}
G_{1}(\varepsilon_{1}+\phi_{1}-\phi_{2})+
G_{2}(\varepsilon_{2}+\phi_{1}-\phi_{2})+
G_{3}(\varepsilon_{3}+\phi_{1}-\phi_{2})=0
\end{align}
from which it easily follows
\begin{align}
\label{v-plane}
G_{1}\varepsilon_{1}+G_{2}\varepsilon_{2}+G_{3}\varepsilon_{3}
=-(G_1+G_2+G_3)(\phi_{1}-\phi_{2})
\quad \text{(voltage plane)}
\end{align}
Hence, the voltage drop $\bar{\de}\phi$ between nodes $v_1,v_2$ is given by
\begin{align}
\label{v-drop}
-\bar{\de}\phi
=
\dfrac{\braket{G|\varepsilon}}{\tr G}
=
\dfrac{\braket{G_1 G_2 G_3|\varepsilon_1\varepsilon_2\varepsilon_3}}{\tr G}
\quad \text{(Millman's formula, $n=3$)}
\end{align}
The latter tells us that for fixed $\bar{\de}\phi$ the algebraic voltages are not fully arbitrary, because the circuit voltage point $V_\varepsilon\=(\varepsilon_1;\varepsilon_2;\varepsilon_3)$ lies on the voltage 2-plane \eqref{v-plane} as the result of the HKCL and HKVL while the circuit branch current point $I_i\=(i_1;i_2;i_3)$ lies on the current 2-plane \eqref{i-plane}. Alternatively, one can consider the Millman's formula \eqref{v-drop} as a generator of $\bar{\de}\phi$ as well.

We know from the HKCL that the \emph{physical} branch currents are cycles, i.e $\p_1\ket{i}=0$. Thus it is natural to search the latter as a boundary  
\begin{align}
\ket{i}\=\p_2\ket{\mu}\quad \in B_1(C)\=\IM\p_2\quad \text{(converse Poincar\'e lemma)}
\end{align}
where $\ket{\mu}$ is called the \emph{mesh current}.  Calculate:
\begin{align}
\ket{i_1i_2i_3}
&=
\begin{bmatrix*}[r]
1 & 0\\
-1 & 1\\
0 & -1
\end{bmatrix*}
\ket{\mu_1\mu_2}
\\
&=
\ket{\mu_1;-\mu_1+\mu_2;-\mu_2}
\end{align}
from which we obtain
\begin{align}
i_1=\mu_1,\quad i_2=-\mu_1+\mu_2,\quad i_3=-\mu_2
\end{align}
and hence the HKCL becomes automatic. The mesh currents may be used to construct bases in $Z_1$.

\section{Homological analysis and modeling}

Now we may collect the homological properties of the electric circuit on Fig. \ref{circuit} as follows.

\begin{thm}[cf \cite{Roth55}]
The electric circuit of Fig. \ref{circuit} can be represented by the following short exact sequence:
\begin{align}
\label{SES2}
\begin{CD}
0 @>   >> C_2 @> \p_2 >> C_1 @> \p_1 >> \IM \p_1 @>   >> 0  
\end{CD}
\end{align}
\end{thm}

\begin{proof}
We compactly collect the basic points of the proof and for convenience use notations of \eqref{e-sequence} as well as the dimensional considerations in \eqref{hom-spaces}.
\begin{enumerate}
\item
Exactness at $C_2$: $0=\Ker\p_2\quad\overset{?}{\Longleftrightarrow}\quad H_2=0$.\\
Note that $\IM\p_3=0$ and  $\rank\p_2=2$ (maximal). Hence, $\dim\IM\p_3=0=\dim\Ker\p_2$, which is equivalent to $H_2=0$.
\item
$\IM\p_2=\Ker\p_1\quad\overset{?}{\Longleftrightarrow}\quad H_1=0$.\\
Note that $\dim\IM\p_2=2=\dim\Ker\p_1$, which is equivalent to $H_1=0$.
\item
$\IM\p_1=\Ker\p_0\quad\overset{?}{\Longleftrightarrow}\quad H_0=0$.\\
Note that $\dim\IM\p_1=1=\dim\Ker\p_0$, which is equivalent to $H_0=0$.
\qedhere
\end{enumerate}
\end{proof}
\begin{cor}
Thus we have
\begin{align}
\IM\p_1\cong\dfrac{C_1}{\Ker\p_1}\cong\dfrac{C_1}{\IM\p_2}
\end{align}
\end{cor}

\begin{rem}[correctness]

As usual in the mathematical physics, by the \emph{correctness} of a (modeling) problem one means:
\begin{enumerate}
\itemsep-2pt
\item
Existence of the solution.
\item
Uniqueness of the solution.
\item
Stability of the solution under the infinitesimal deformation of the physical parameters.
\end{enumerate}
Note that  conditions 1 and 2 are related to the short exact sequence \eqref{SES2} while the explicit form of the solution and its stability is given by the Millman formula \eqref{v-drop}. For more careful study one has to apply the \emph{cohomological} analysis as well, the latter can be realized by the \emph{dual} to the short exact sequence  \ref{SES2}. Here, we omit the cohomological analysis, one can find more details in  \cite{Roth55,BS90}.

\end{rem}

\section{Numerical example}

As a simple example determine the branch and mesh currents in the circuit shown in Fig. \ref{circuit}. Take the electric parameters as (we follow \cite{O'Malley}, Problem 4.10):
\begin{align}
&\varepsilon_{1}\,\,=40V,\quad \varepsilon_{2}=12V,\quad\varepsilon_{3}=-24V\\
& R_{1}=6\Omega,\quad\, R_{2}=4\Omega,\quad\, R_{3}=12\Omega
\end{align}
First calculate the voltage drop
\begin{align}
\bar{\de}\phi
&\=
\phi_{1}-\phi_{2}\\
&=
-\dfrac{40V\cdot4\Omega\cdot12\Omega+12V\cdot6
\Omega\cdot12\Omega-24V\cdot6\Omega\cdot4\Omega}{6\Omega\cdot4\Omega+4\Omega\cdot12\Omega+12\Omega\cdot6\Omega}\\
&=
-\dfrac{(1920+864-576)V\Omega^{2}}{(24+48+72)\Omega^{2}}\\
&=
-\dfrac{184}{12}V
\end{align}
Then calculate the circuit and mesh currents
\begin{align}
i_{1}
&=\dfrac{\varepsilon_{1}+\bar{\de}\phi}{R_{1}}
=
\dfrac{12\cdot40V-184V}{6\Omega\cdot12}=\dfrac{296V}{72\Omega}
=\dfrac{148}{36}A
&&=\mu_1
\\
i_{2}
&=
\dfrac{\varepsilon_{2}+\bar{\de}\phi}{R_{2}}
=
\dfrac{12\cdot 12V-184V}{4\Omega\cdot12}=-\dfrac{40V}{48\Omega}=
-\dfrac{40V\cdot\frac{3}{4}}{48\Omega\cdot\frac{3}{4}}
=-\dfrac{30}{36}A
&&=-\mu_1+\mu_2
\\
i_{3}
&=
\dfrac{\varepsilon_{3}+\bar{\de}\phi}{R_{3}}
=\dfrac{12\cdot (-24V)-184V}{12\Omega\cdot12}=\dfrac{-472V}{144\Omega}
=-\dfrac{118}{36}A
&&=-\mu_2
\end{align}
Finally, check the KCL
\begin{align}
i_{1}+i_{2}+i_{3}=\dfrac{148-30-118}{36}A=0
\end{align}
Computer simulation for this particular circuit can easily be arranged.

\section*{Acknowledgements}

The research was in part supported by the Estonian Research Council, Grant ETF-9038. Authors are grateful to participants of MOD II (Moduli, Operads, Dynamics II, Tallinn, 03-06 June 2014) and, in particular, to A. Siqveland and A. A. Voronov for the thorough discussions that helped to improve the quality of the manuscript.

\noindent
Tallinn University of Technology, Estonia

\end{document}